\documentclass[a4paper,10pt]{article}
\usepackage{amsmath, amssymb}
\usepackage{graphicx}
\usepackage{color}
\usepackage{graphicx}
\usepackage{url}
\usepackage{color,soul}
\usepackage{float}
\usepackage{amsmath}
\usepackage{amssymb}
\usepackage{courier}
\usepackage{multicol}
\usepackage[abs]{overpic}
\usepackage{enumitem}
\usepackage{multicol}
\usepackage{amsthm}
\usepackage{changepage}
\usepackage{framed}
\usepackage{mdframed}
\usepackage{lipsum}
\usepackage{longtable}
\newmdenv[
linecolor=gray,
linewidth=3pt,
topline=false,
bottomline=false,
rightline=false,
rightmargin=-10pt,
leftmargin=0pt,
skipabove=\topsep,
skipbelow=\topsep
]{siderules}

\newtheorem{theorem}{Theorem}
\newtheorem{lemma}[theorem]{Lemma}
\newtheorem{claim}[theorem]{Claim}

\usepackage[linesnumbered,boxed]{algorithm2e}
\allowdisplaybreaks
\usepackage{pgf,tikz}
\usepackage{mathrsfs}
\usetikzlibrary{arrows}
\usepackage{times}
\DeclareMathSizes{10}{9}{6}{5}

\def\EE{\mathbb E}

\def\RR{\mathbb R}

\def\ZZ{\mathbb Z}

\def\mC{\mathscr C}
\def\cA{\mathcal A}
\def\cB{\mathcal B}
\def\cC{\mathcal C}
\def\cD{\mathcal D}

\def\cS{\mathcal S}
\def\cP{\mathcal P}

\def\cQ{\mathcal Q}

\def\cX{\mathcal X}

\def\bC{\mathbf C}

\def\bE{\mathbf E}

\def\b1{\mathbf 1}

\def \Ancrr {\textsc{Ancrr}}
\def \OPT {\textsc{Opt}}

\newcommand{\poly}{\operatorname{poly}}
\newcommand{\polylog}{\operatorname{polylog}}

\newcommand{\argmin}{\operatorname{argmin}}
\newcommand{\argmax}{\operatorname{argmax}}

\newtheorem{fact}{Fact}

\newtheorem{remark}{Remark}
 
\newtheorem{definition}{Definition}
\newtheorem{corollary}{Corollary}
\newcommand{\raf}[1]{(\ref{#1})}

\usepackage[top=2.11cm, bottom=1.8cm, left=1.8cm, right=1.8cm]{geometry}

\title{Some Black-box Reductions for Objective-robust Discrete Optimization Problems Based on their LP-Relaxations}
\author{
	Khaled Elbassioni\thanks{Khalifa University of Science and Technology, Masdar City Campus, P.O. Box 54224, Abu Dhabi, UAE;
		(kelbassioni@masdar.ac.ae)}
}
\date{}
\begin{document}
\maketitle

\begin{abstract}
We consider robust discrete minimization problems where uncertainty is defined by a convex set in the objective.  We show how an integrality gap verifier for the  linear programming relaxation of the non-robust version of the problem can be used to derive approximation algorithms for the robust version.
\end{abstract}
\section{Introduction}\label{Intro}

Standard optimization algorithms assume precise knowledge of their inputs, and find optimal or near-optimal solutions under this assumption. However, in real-life applications, the input data  may be known up to a limited precision with errors introduced possibly due to inaccuracy in measurements or lack of exact information about the precise value of the input parameters. Clearly, an optimization algorithm designed based on such distorted data to optimize a certain objective function would not yield reliable results, if no special consideration of such uncertainty is taken. 
Several approaches to deal with uncertainty in data have been introduced, including {\it stochastic optimization} (see e.g., \cite{BL11}), where certain probabilistic assumptions are made on the uncertainty and the objective is optimize the average-case or the probability of a certain desirable event, and {\it robust optimization} (see, .eg., \cite{BEN09}), where some deterministic assumptions are made on the uncertain parameters, and the objective is to optimize over the worst-case  these parameters can assume\footnote{Yet, there is a third (intermediate) approach, namely, {\it distributionally robust optimization} (see, e.g., \cite{DY10}), in which one optimizes the expectation over the worst-case choice from a set of distributions on the uncertain parameters.}.   

In this paper, we consider a class of {\it robust} discrete optimization (DO) problems, where uncertainty is assumed to be only {\it in the objective} (called sometimes {\it cost-robust} optimization problems). Given a discrete set of solutions, on is interested in maximizing/minimizing a linear objective function over this set; it is assumed that the objective function is not explicitly given, but is known to belong to a {\it convex} an uncertainty set. The requirement is to solve the optimization problem in the {\it worst-case} scenario that the objective assumes in the uncertainty set.
 Our goal is to show how an approximation algorithm, based on the linear programming (LP) relaxation for the nominal version of a discrete optimization problem, can be used to derive an approximation algorithm for the robust version. We will focus on minimization problems, even though some of the results can be extended to maximization problems. 

\subsection{Integrality Gap Verifiers}
More formally, we consider a minimization problem over a discrete set $\cS\subseteq\ZZ^n$ and a corresponding LP-relaxation over $\cQ\subseteq\RR^n_+$:

{\centering \hspace*{-18pt}
	\begin{minipage}[t]{.47\textwidth}
		\begin{alignat}{3}
		\label{cop-1}
		\quad& \displaystyle \OPT = \min\quad c^T x\\
		\text{s.t.}\quad & \displaystyle x\in \cS\nonumber
		\end{alignat}
	\end{minipage}
	\,\,\, \rule[-10ex]{1pt}{10ex}
	\begin{minipage}[t]{0.47\textwidth}
		\begin{alignat}{3}
		\label{lprlx}
	\quad& \displaystyle z^* = \min\quad c^Tx\\
		\text{s.t.}\quad & \displaystyle x\in\cQ\nonumber,
		\end{alignat}
\end{minipage}}

\noindent where $c\in\RR_+^n$. 
We will be mainly working with discrete optimization problems for which there is an approximation algorithm that rounds any feasible LP solution to a discrete one with a bounded approximation ratio. This is formulated in the following definition.   
\begin{definition}\label{d1}
	For $\alpha\ge 1$, a (deterministic) $\alpha$-{\it integrality gap verifier} $\cA=\cA(c,x)$ for~\raf{cop-1}-\raf{lprlx}, w.r.t. a class $\mC\subseteq\RR_+^n$ of objectives is a polytime algorithm that, given any $c\in\mC$ and any $x\in\cQ$ returns an $\widehat x\in\cS$ such that $c^T\hat x\le\alpha\cdot c^Tx$. An integrality gap verifier $\cA$ is said to be \emph{ oblivious} (see, e.g., \cite{FFT16}) if $\cA(c,x)=\cA(x)$ does {\it not} depend on the objective $c$. When the the class of objectives is $\mC=\RR_+^n$, we simply call $\cA$ an (oblivious) integrality gap verifier.
\end{definition}

A {\it randomized} $\alpha$-{\it integrality gap verifier} is the same as in Definition~\ref{d1} except that it returns a random $\widehat x\in\cS$ such that $\EE[c^T\widehat x]\le\alpha c^Tx$. We will consider a special class of randomized integrality gap verifiers that are given by the following definition.

\begin{definition}\label{d2}
	For $\alpha\ge 1$ and $x\in\cQ$, an $\alpha$-{\it approximate (semi-) negatively correlated randomized rounding}, denoted $\alpha$-\textsc{Ancrr}, of $x$ is an $\widehat x\in\cS$ such that:
	\begin{itemize}
		\item [(i)] $\EE[c^T\widehat x]\le\alpha c^Tx$; 
		\item [(ii)] For any $S\subseteq[n]$:
		\begin{align}
		\Pr\Big[\bigwedge_{i\in S}(\widehat x_i=1)\Big]&\le\prod_{i\in S}\Pr[\widehat x_i=1].\label{nc2}
		\end{align} 		
	\end{itemize}
An $\alpha$-\Ancrr\ integrality gap verifier is a polytime algorithm that, given any $x\in\cQ$, returns an $\alpha$-\Ancrr. 
\end{definition}

\begin{remark}\label{r0}
Consider a minimization problem \raf{cop-1} and its LP relaxation~\raf{lprlx}. By Markov's inequality, given an $\alpha$-randomized intergality gap verifier $\cA$, $x\in\cQ$, $c\in\RR_+^n$ and $\epsilon>0$, we can get in $\poly(n,\frac{1}{\epsilon})$ calls to $\cA$ an $\widehat x\in\cS$ such that $c^T\widehat x\leq(1+\epsilon)\alpha\cdot c^Tx$ holds with probability $1-o(1)$.
\end{remark}

\subsection{Example: {\sc SetCover}}
Let $V$ be a finite set of $m$ elements. Given sets $S_1,\ldots,S_n\subseteq V$, with 
{\it non-negative costs} $c_1,\ldots,c_n$, the objective is to find a {\it minimum-cost} selection of sets that {\it covers} all the elements of $V$.
The problem and its standard LP relaxation are given as follows:
	{\centering
		\begin{minipage}[t]{.43\textwidth}
			\begin{alignat*}{3}
			\quad& \displaystyle \OPT=  \min\quad \sum_{i=1}^nc_ix_i\\
			\text{s.t.}\quad & \displaystyle \sum_{i:~j\in S_i} x_i\ge 1,\quad \forall j\in V\nonumber\\
			\quad&x_i\in\{0,1\},\quad \forall i\in [n]
			\end{alignat*}
		\end{minipage}
		\,\,\, \rule[-20ex]{1pt}{20ex}
		\begin{minipage}[t]{0.43\textwidth}
			\begin{alignat*}{3}
			\quad& \displaystyle z^* =  \min\quad \sum_{i=1}^nc_ix_i\\
			\text{s.t.}\quad & \displaystyle \sum_{i:~j\in S_i} x_i\ge 1,\quad \forall j\in V\nonumber\\
			\quad&x_i\ge0,\quad \forall i\in [n]
			\end{alignat*}
	\end{minipage}}

It is well-known that the {\it greedy algorithm} that repeatedly picks a set with minimum cost to number of newly covered elements-ratio (deterministically) verifies an integrality gap of $O(\log m)$ for the standard LP relaxation of {\sc SetCover}. Moreover, let $x^*$ be the {\it fractional} optimal solution (to the LP relaxation). Then it is also well-known that the algorithm that picks each set $S_i$ {\it independently} with probability $\min\big\{6x^*_i\log m,1\big\}$ is an $O(\log m)$--\Ancrr\ integrality gap verifier.

\subsection{Robust Discrete Optimization Problems}

In the framework of robust optimization (see, .e.g. \cite{BEN09,BN02}), we assume that the objective vector $c$ is {\it not known exactly}. Instead, it is given by a {\it convex uncertainty set} $\cC\subseteq\RR^{n}_+$.
It is required to find a (near)-optimal solution for the DO problem under the {\it worst-case} choice of objective $c\in\cC$.  Typical examples of uncertainty sets $\cC$ include:
	\begin{itemize}
		\item Polyhedral uncertainty: $\cC:=\cP(A,b,c^0):=\{c\in\RR^{n}_+:A(c-c^0)\le b\}$, for given matrix $A\in\RR^{m\times n}_+$, vector $b\in\RR^{m}_+$ and (nominal) vector $c^0\in\RR^{n}_+$.		
		\item Ellipsoidal uncertainty: $\cC:=\bE(c^0,D):=\{c\in\RR^{n}_+:(c-c^0)^TD^{-2}(c-c^0)\le 1\}$, for given positive definite matrix $D\in\RR^{m\times n}$ and vector $c^0\in\RR^{n}_+$.
	\end{itemize}
More generally, we will consider a class of uncertainty sets defined by {\it affine perturbations} around a nominal vector $c^0\in\RR_+^n$ (see,. e.g., \cite{BEN09}): 
\begin{align}\label{Affine}
\cC=\cC(c^0,c^1,\ldots,c^r;\cD):=\left\{c:=c^0+\sum_{r=1}^k\delta_rc^r:~\delta=(\delta_1,\ldots,\delta_k)\in\cD\right\}, 
\end{align}
where  
$c^0,c^1,\ldots,c^k\in\RR^n_+$ and $\cD\subseteq\RR^k$ is a {\it convex} perturbation set:
\begin{itemize}
	\item Polyhedral perturbation $\cD=\cP(A,b,0):=\{\delta\in\RR^{k}_+:A\delta\le b\}$, for given matrix $A\in\RR^{m\times k}_+$ and vector $b\in\RR^{m}_+$.		
	\item Ellipsoidal perturbation: $\cD=\bE(0,D):=\{\delta\in\RR^{k}_+:\delta^TD^{-2}\delta\le 1\}$, for a given positive definite matrix $D\in\RR^{k\times k}$.
\end{itemize}
The vectors $c^1,\ldots,c^k\in\RR^n_+$ are called the {\it generators} of the perturbation set $\cD$. 
Note that a polyhedral uncertainty set $\cP(A,b,c^0)$ can be described in the form~\raf{Affine} by setting $\cC:=\cC(c^0,\b1^1_n,\ldots,\b1^n_n;\cD)$ for the polyhedral perturbation set $\cD:=\cP(A,b,0)$, where $\b1^j_n$ denotes the $j$th unit vector in $\RR^n$. Similarly, an ellipsoidal uncertainty set $E(c^0,D)$ can be described in the form~\raf{Affine} by setting $\cC:=\cC( c^0,\b1^1_n,\ldots,\b1^n_n;\cD)$ for the ellipsoidal perturbation set $\cD:=\bE(0,D)$.
\subsection{Convex Relaxation for the Robust DO Problem}
We can model the robust DO problem and its convex relaxation as follows:

{\centering
	\begin{minipage}[t]{.47\textwidth}
		\begin{alignat}{3}
		\label{rcop-1}
		\quad& \displaystyle \OPT_R =\min_{x\in \cS} \max_{c\in\cC}\quad c^Tx,
		\end{alignat}
	\end{minipage}
	\,\,\, \rule[-8ex]{1pt}{7ex}
	\begin{minipage}[t]{0.47\textwidth}
		\begin{alignat}{3}
			\label{rlp-rlx}
		\quad& \displaystyle z_R^* =\min_{x\in \cQ} \max_{c\in\cC}\quad c^Tx.
		\end{alignat}
\end{minipage}}

Equivalenlty, we can write~\raf{rcop-1}-\raf{rlp-rlx} as

{\centering
	\begin{minipage}[t]{.47\textwidth}
		\begin{alignat}{3}
		\label{ercop-1}
		\quad& \displaystyle \OPT_R = \min\quad z \\
		\text{s.t.}\quad & \displaystyle c^Tx\le z \quad \forall {c\in\cC}\label{ercop-1-e1}\\
		&x\in \cS.\nonumber
		\end{alignat}
	\end{minipage}
	\,\,\, \rule[-14ex]{1pt}{14ex}
	\begin{minipage}[t]{0.47\textwidth}
		\begin{alignat}{3}
		\label{erlp-rlx}
		\quad& \displaystyle z_R^* =  \min\quad z\\
		\text{s.t.}\quad & \displaystyle c^Tx\le z \quad \forall {c\in\cC}\label{erlp-rlx-e1}\\
		& x\in\cQ.\nonumber
		\end{alignat}
\end{minipage}}

Note that~\raf{rlp-rlx} amounts to a convex programming problem that can be solved (almost to optimality) in polynomial time (see, e.g., \cite{GLS93}). {\it Near-optimal} solutions can also be found more {\it efficiently},  based on the semi-infinite LP formulation~\raf{erlp-rlx}, using the {\it multiplicative weight updates method}  \cite{EMN19}. 

\subsection{Guarantees for a Robust DO problem}

We consider both deterministic and randomized algorithms for the robust optimization problem (see, e.g., \cite{BS03,KS18}): 
\begin{definition}
	For $\alpha\ge 1$, a {\it randomized approximation algorithm} $\cB$ for the robust DO problem~\raf{rcop-1} is said to be:
	\begin{itemize}
		\item $\alpha$-\emph{robust-in-expectation} (w.r.t. the uncertainty set $\cC$), if the expected objective in the uncertainty set $\cC$,  w.r.t. the output solution, over the random choices of the algorithm, is within a factor of $\alpha$ from the optimum solution:
	\begin{align*}
	\EE_{\widehat x\in\cB}[c^T\widehat x]\le \alpha \cdot \OPT_R\quad\forall c\in\cC;
	\end{align*}
	\item $\alpha$-\emph{robust-with-high-probability}, if with probability  approaching $1$, all objectives in the uncertainty set $\cC$, w.r.t. the output solution, are within a factor of $\alpha$ from the optimum solution:
	\begin{align*}
\Pr_{\widehat x\in\cB}[c^T\widehat x\le \alpha \cdot \OPT_R\quad\forall c\in\cC]=1-o(1);
	\end{align*}
	\item $\alpha$-\emph{deterministically robust} if it is $\alpha$-robust {\it with probability $1$}, i.e., it outputs a vector $\widehat x\in\cS$ such that:
	\begin{align*}
	 c^T\widehat x\le \alpha \cdot \OPT_R\quad\forall c\in\cC.
	\end{align*}			
	\end{itemize}
\end{definition}
Clearly, the notion of $\alpha$-deterministically robust is stronger than that of $\alpha$-robust-with-high-probability, which is, in turn, (more or less) stronger than that of $\alpha$-robust-in-expectation.

\subsection{Summary of Main Results}

To describe the results we obtain in this paper, let us consider the polyhedral/ellipsoidal uncertainty sets:
		\begin{align}
		\cC_1&:=\left\{c:=c^0+u~\left|~u\in\RR_+^n,~u\le d,\quad Au\le b\right.\right\}\label{cC1}\\
		\cC_2&:=\left\{c:=c^0+\sum_{r=1}^k\delta_rc^r~\left|~\delta\in\RR_+^k,~A\delta\le b\right.\right\}\label{cC2}\\
		\cC_3&:=\left\{c:=c^0+u~\left|~u\in\RR_+^n,~\|D^{-1}u\|_2\le 1\right.\right\}\label{cC3}\\
		\cC_4&:=\left\{c:=c^0+\sum_{r=1}^k\delta_rc^r~\left|~\delta\in\RR_+^k,~\|D^{-1}\delta\|_2\le 1\right.\right\}\label{cC4}
		\end{align}
   Assume $A,b,d,\bC$ are non-negative and $D$ is positive definite, where $\bC\in\RR_+^{n\times k}$ is the matrix whose columns are $c^1,\ldots,c^k$.
Let $m$ be the number of rows of $A$, $\beta:=\min_j \max_i a_{ij}$ and $\gamma:=\max_{i,j}a_{ij}$, $c_{\min}:=\min_{r\ne 0,j:~c^r_{j}>0}c^r_{j}$ and $c_{\max}:=\max_{r\ne 0,j}c^r_{j}$. Our results are summarized in Table~\ref{tab1}. The first column describes the restrictions on the discrete set $\cS$ (if any): $\cS$ is {\it binary} if $\cS\subseteq\{0,1\}^n$ and {\it covering}  if  $x\in\cS$ and $y\ge x$ implies $y\in\cS$. In the second column, we describe the type of uncertainty set considered, and the conditions on it (if any). The third column  gives the type of approximation algorithm which we assume available for the nominal problem, while the fourth column gives the  guarantee for the corresponding robust version. As can be seen from the table, except for the first two results, the approximation factors we obtain depend on the "width" of the uncertainty set as described by the ratios $\frac{\gamma}{\beta}$ and $\frac{c_{\max}}{c_{\min}}$ for polyhedral uncertainty, and  $\frac{\lambda_{\max}(D)}{\lambda_{\min}(D)}$ for ellipsoidal uncertainty.  The approximation ratio is also proportional to the square root of the number of generators in the perturbation set. Whether these bounds can be significantly improved remains an interesting open question.  

\begin{table}[h]
	\begin{center}
	\begin{tabular}{|l|l|l|l|}
		\hline
		$\cS$& Uncertainty set & Available black-box & Approximation guarantee\\
		\hline\hline
		General& general convex set& general $\alpha$-integrality &$\alpha$-robust-in-expectation\\
		& & gap verifier& \\\hline
		Binary& $\cC_1$; $m=O(1)$ & general $\alpha$-approx. Alg.& $O(\alpha)$-deterministically robust\\\hline
		Binary \&  & $\cC_1$ & general $\alpha$-integrality  & $O\big(\alpha+\sqrt{\frac{\alpha\gamma n}{\beta}}\big)$-\\
		covering & & gap verifier& deterministically robust\\\hline
		Binary & $\cC_2$ & $\alpha$-\Ancrr\ integrality & $O\big(\alpha\sqrt{k \log(k)\frac\gamma\beta\frac{c_{\max}}{c_{\min}}}\big)$-\\
		& & gap verifier& robust-with-high-probability \\\hline
		General & $\cC_4$; $D\bC\ge0$ & general $\alpha$-integrality & $O\big(\alpha\sqrt{k}\big)$-\\
		& & gap verifier& deterministically robust- \\\hline
		Binary \&  & $\cC_3$; $D^{-1}>0$  & general $\alpha$-integrality  & $O\big(\alpha+\sqrt{\frac{\alpha\lambda_{\max}(D) n}{\lambda_{\min}(D)}}\big)$-\\
		covering & & gap verifier& deterministically robust\\\hline
		Binary & $\cC_4$; $D^{-1}>0$ & $\alpha$-\Ancrr\ integrality & $O\big(\alpha\sqrt{k \log(k)\frac{\lambda_{\max}(D)}{\lambda_{\min}(D)}\frac{c_{\max}}{c_{\min}}}\big)$-\\
		& & gap verifier& robust-with-high-probability \\\hline
	\end{tabular}
\end{center}
	\label{tab1}
\caption{Summary of the reductions.}
\end{table}

\subsection{Some Related Work}
While there is an extensive body of work on robust {\it continuous} optimization problems (see, e.g., \cite{BEN09,BN98,BN02,BBC11,BS06,EOL98,P16}), much less is known in the {\it discrete} case, where most work has considered special uncertainty sets or specific discrete problems. In \cite{BS03}, Bertsimas and Sim consider the minimization problem~\raf{rcop-1} with {\it budget uncertainly}, where at most $k$ components of the objective are allowed to increase; for {\it binary} optimization problems they gave an $\alpha$-deterministically robust approximation algorithm for the robust version which is obtained by making $n+1$ calls to any $\alpha$-approximation algorithm for the non-robust version.  Some generalizations of this result to the {\it non-binary} case were obtained in \cite{GSTP12}, and other improvements and generalizations were obtained  in~\cite{ALT13}. In Section~\ref{budget} below, we show that the number of calls to the approximation algorithm can be made significantly smaller and also extend the result to any {\it constant} number of budget constraints. For {\it uncorrelated ellipsoidal uncertainty} (where the uncertainty set is an axis-aligned ellipsoid), Bertsimas and Sim \cite{BS04} also gave a {\it pseudo polynomial-time} reduction from solving a robust version problem over a binary set $\cS$ to a linear optimization problem over the same set. As observed in \cite[Chapter 2]{I17}, when specialized to {\it ball uncertainty}, this yields a polynomial time algorithm for solving the robust problem, whenever the nominal version can be solved in polynomial time. This should be contrasted with our result in Theorem~\ref{PU3-thm}, where an $O(\alpha\sqrt{n})$-approximation for the robust problem  with ellipsoidal uncertainty, satisfying $D>0$, over an arbitrary discrete set, can be obtained from any $\alpha$-integrailty gap verifier for the nominal problem.
 
More recently, Kawase and Sumita (2018) gave robust-in-expectation algorithms for special problems such as the {knapsack} problem and the {maximum independent set} problem in the {\it intersection of $r$ matroids}, among others. We note, however, that their results are not of the black-box type, that is, they provide algorithms that are specific to each problem. We note also that some of these results can be derived from our reduction in Section~\ref{exp}.  Finally, it is worth noting that there is a number of results on special problems, such as {\sc ShortestPath} \cite{PPGP15}, {\sc MinCostFlow} \cite{BS03}, 
{\sc MachineScheduling} \cite{BPP19}, {\sc VehicleRouting }\cite{ASNP16}, two-stage robust optimization \cite{DGRS05,FJMM07}, mostly under a class of budget uncertainty.  In general, this seems to be a growing area of research, see, e.g., the theses by Poss~\cite{P16} and Ilyina~\cite{I17}.

\medskip

\noindent{\bf Outline of the techniques.} All the results in Table~\ref{tab1} are based on solving the convex relaxation for the robust  optimization problem (in some form or th other), then rounding the obtained fractional solution. A useful tool that we rely on, first proved by Carr and Vempala \cite{CV02}, allows one to use a given integrality gap verifier for the LP-relaxation to round the fractional solution without losing much in the objective. Another ingredient of our proofs is the use of {\it strong LP-duality} to go from a $\text{maxmin}$-optimization problem to a purely minimization problem; this was the approach used by Bertsimas and Sim in \cite{BS03}, which we push further by combining it with randomized rounding techniques, and using a {\it dual -fitting} argument to bound the approximation guarantee on the rounded solution. First we describe this approach for polyhedral uncertainty, then it would not be hard to extend the results to ellipsoidal uncertainty, by envisioning an ellipsoid as a polytope with {\it infinitely many} linear inequalities. 

\section{A Robust-in-Expectation Approximation Algorithm}\label{exp}


We first observe simply that an {\it oblivious} intergality gap verifier for the nominal problem implies an $\alpha$-robust-in-expectation algorithm for the robust version.
\begin{lemma}\label{t-rie-lp}
	Consider a combinatorial  minimization problem~\raf{cop-1} and its LP relaxation~\raf{lprlx}, admitting an oblivious $\alpha$-integrality gap verifier $\cA$ w.r.t. a class $\mC$ of objectives. Then there is a polytime $\alpha$-robust-in-expectation algorithm for the robust version~\raf{ercop-1} w.r.t. to the any convex uncertainty set $\cC\subseteq\mC$.
\end{lemma}
\begin{proof}
	We solve the robust convex relaxation~\raf{erlp-rlx} to find $z^*_R$ and a corresponding optimal solution $x^*\in\cQ$.
	Since $\cA$ is an oblivious $\alpha$-integrality gap verifier, we have for all $c\in\cC$,
	\begin{align*}
	\EE_{\widehat x\sim\cA(x^*)}[c^T\widehat x]&\le\alpha \cdot c^Tx^*\le \alpha \cdot z^*_R\le \alpha \cdot \OPT_R,
	\end{align*}  
	where the second inequality follows by~\raf{erlp-rlx-e1}.
\end{proof}

Carr and Vempala \cite{CV02} gave a decomposition theorem that allows one to use an $\alpha$-integrality gap verifier for a given LP-relaxation of a combinatorial minimization problem, to decompose a given fractional solution to the LP into a convex combination of integer solutions that is dominated by $\alpha$ times the fractional solution. We can restate their result as follows.
\begin{theorem}[\cite{CV02}]\label{CV-LP}
	Consider a discrete  minimization problem~\raf{cop-1} and its LP relaxation~\raf{lprlx}, admitting an $\alpha$-integrality gap verifier $\cA$. Then there is a polytime algorithm that, for any given $x^*\in\cQ$, finds a set $\cX\subseteq\cS$, of {\it polynomial} size, and a set of convex multipliers $\{\mu_x\in\RR_+:~x\in\cX\}$, $\sum_{x\in\cX}\mu_x=1$,  such that 
	\begin{align}\label{cvx-comb-LP}
      \alpha x^*\geq\sum_{x\in\cX}\mu_xx.
	\end{align}
\end{theorem}     
We obtain the following (known) corollary of Theorem~\ref{CV-LP}, whose proof is included for completeness.
\begin{corollary}\label{c1}
	Consider a discrete minimization problem~\raf{cop-1} and its LP relaxation~\raf{lprlx}, admitting an $\alpha$-integrality gap verifier $\cA$. Then~\raf{lprlx} admits an oblivious $\alpha$-integrality gap verifier $\cA'$.
\end{corollary}
\begin{proof}
	Given a ({\it non-oblivious}) $\alpha$-integrality gap verifier $\cA$ for ~\raf{lprlx}, we can construct a randomized {\it oblivious} $\alpha$-integrality gap verifier $\cA'$ as follows.  By Theorem~\ref{CV-LP}, for any given $x^*\in\cQ$ we can get a dominated convex combination as in \raf{cvx-comb-LP}, with a polynomially sized set $\cX:=\{x\in\cS:~\mu_x>0\}$.  As $\sum_{x\in\cX}\mu_x=1$, these convex multipliers define a {\it probability distribution} over $\cX$. Let $\widehat x\in\cX$ be selected according to this distribution. Then  
	$\EE[\widehat x]=\sum_{x\in\cX}\mu_xx\leq\alpha x^*$
		 by \raf{cvx-comb-LP}. It follows, by linearity of expectation, that for any $c\in\RR_+^n$, we have 
	 \begin{align*}
	 \EE[c^T\widehat x]=c^T\EE[\hat x]\le\alpha c^Tx^*.
	 \end{align*}
\end{proof}

From Lemma~\ref{t-rie-lp} and Corollary~\ref{c1}, we obtain an $\alpha$-robust-in-expectation algorithm for~\raf{rcop-1} from an $\alpha$-integrality gap verifier for~\raf{cop-1}-\raf{lprlx}.

\begin{theorem}\label{t-main-lp}
	Consider a discrete minimization problem~\raf{cop-1} and its LP relaxation~\raf{lprlx}, admitting an $\alpha$-integrality gap verifier $\cA$. Then there is a polytime $\alpha$-robust-in-expectation algorithm for the robust version~\raf{ercop-1} w.r.t. to the any convex uncertainty set $\cC\subseteq\RR_+^n$.
\end{theorem}
 We emphasize that, in Theorem~\ref{t-main-lp}, the integrality gap verifier must be defined with w.r.t. the {\it whole} class $\mC=\RR_+^n$ of objectives. Finally, we note that the results in this section can be extended, with no difficulty, to maximization problems.

\subsection{A Deterministically Robust Algorithm for a Class of Polyhedral Uncertainty}\label{budget}

In \cite{BS03}, Bertsimas and Sim considered the minimization version of the DO problem \raf{cop-1}, when the set $\cS\subseteq\{0,1\}^n$ and the (budget) uncertainty set $\cC$ is given by 
\begin{align}\label{U1}
\cC=\left\{c:=c^0+d\circ u~\left|~u\in\RR_+^n,~u_i\le 1, \forall i\in[n],\quad \sum_{i=1}^nu_i\le k\right.\right\},
\end{align}
where $c^0,d\in\RR_+^n$ are given non-negative vectors, $k\in\ZZ_+$ is a given positive integer, and $d\circ y$ is the $n$-dimensional vector with components $(d\circ u)_i:=d_iu_i$, for $i=1,\ldots,n$. The constraints in \raf{U1} describe the situation when the uncertainty in each component of the objective vector $c$ is described by an interval $[c^0_j,c^0_j+d_j]$ and at most $k$ components are allowed to change. It was shown in \cite{BS03} that an $\alpha$-deterministically robust approximation algorithm for the minimization version of \raf{rcop-1} with the uncertainty set given in \raf{U1}, can be obtained from $n+1$ calls to an $\alpha$-approximation algorithm for the nominal problem \raf{cop-1}.

In this section, we extend this result as follows. Consider a polyhedral uncertainty set given by
	\begin{align}\label{U2}
\cC=\left\{c:=c^0+u~\left|~u\in\RR_+^n,~u\le d,\quad Au\le b\right.\right\},
\end{align}
where $d\in\RR_+^n$, $b\in\RR_+^m$ are given non-negative vectors and $A\in\RR_+^{m\times n}$ is given non-negative matrix.
Note that the uncertainty set $\cC$ in \raf{U1} can be written in the form 
\raf{U2} by replacing $d\circ u$ by $u$ and setting $A:=\left[\begin{array}{lll}
\frac1{d_1}&\cdots&\frac1{d_n}\end{array}\right]\in\RR_+^{1\times n}$, $b:=\left[\begin{array}{l}k\end{array}\right]\in\RR_+^1$ (assuming w.l.o.g. that $d_i>0$ for all $i$).

Fix an $\epsilon>0$. As we shall see below, we may assume, w.l.o.g., that $b_i>0$ for all $i\in[m]$. Define
\begin{align}\label{LAcd}
L(A,c^0,d):=n\cdot\max\left\{\frac{\max_jc^0_j}{\min_{j}c^0_j},\frac{(m+n)}{\epsilon}\cdot\max\Big\{\frac{\max_{i,j}a_{ij}/b_i}{\min_j\max_ia_{ij}/b_i},\frac{\max_jd_j}{\min_jd_j}\Big\}\right\}.
\end{align}
\begin{theorem}\label{PU-thm}
	Consider the DO problem \raf{cop-1}, when the set $\cS\subseteq\{0,1\}^n$ and the uncertainty set $\cC$ is given by \raf{U2}. Then, for any given $\epsilon>0$, there is an $\alpha$-deterministically robust approximation algorithm for the cost-robust version~\raf{rcop-1}, which can be obtained from $O(\frac{\log L(A,c^0,d)}{\epsilon} (\log\frac{(1+\epsilon)m}{\epsilon})^m)$ calls to an $\alpha$-integrality gap verifier for the nominal problem \raf{cop-1}. 
\end{theorem}
\begin{proof}
	Assume the availability of an $\alpha$-integrality gap verifier for the nominal problem \raf{cop-1}.  We assume w.l.o.g. that $0\not\in\cS$.
    Note that the robust DO problem~\raf{rcop-1} in this case takes the form: 
    \begin{alignat}{3}
    \label{r-1}
    \quad& \displaystyle \OPT_{R} = \min_{x\in \cS}\left\{(c^0)^Tx+ \max_{u\in\RR_+^n:~u\le d,~ Au\le b}x^Tu\right\}.
    \end{alignat}
    Let us consider the inner maximization problem in~\raf{r-1} and its dual (for a given $x\in\RR_+^n$):
    
    {\centering \hspace*{-18pt}
    	\begin{minipage}[t]{.47\textwidth}
    		\begin{alignat}{3}
    		\label{ir-1}
    		\quad& \displaystyle z^*(x) = \max \quad x^T u\\
    		\text{s.t.}\quad & \displaystyle Au\le b, \label{Aub}\\ 
    		\quad& u\le d,~
    		 u\in\RR_+^n\nonumber
    		\end{alignat}
    	\end{minipage}
    	\,\,\, \rule[-14ex]{1pt}{14ex}
    	\begin{minipage}[t]{0.47\textwidth}
    		\begin{alignat}{3}
    		\label{ir-1-d}
    		\quad& \displaystyle z^*(x) = \min \quad b^T \theta+d^Ty\\
    		\text{s.t.}\quad & \displaystyle A^T\theta+y\ge x, \label{y}\\
    		\quad& \theta\in\RR_+^m,~y\in\RR_+^n.\label{nong}
    		\end{alignat}
    \end{minipage}}

\noindent Following \cite{BS03}, we write~\raf{r-1} using the dual~\raf{ir-1-d} to obtain
\begin{alignat}{3}
\label{r-2}
\quad& \displaystyle \OPT_{R} = \min \quad (c^0)^Tx+b^T \theta+d^Ty\\
\text{s.t.}\quad & \displaystyle A^T\theta+y\ge x, \label{cover}\\
\quad& \theta\in\RR_+^m,~y\in\RR_+^n,\nonumber\\
\quad& x\in \cS.\nonumber
\end{alignat}
Let $a^j\in\RR_+^m$ denote the $j$th column of $A$.
The high-level idea of the approximation algorithm for \raf{r-1} is as follows.
Suppose we know the minimizer $\theta^*$ in \raf{r-2}. Then by \raf{y} and \raf{nong}, $y_j=\max\{x^*_j-(a^j)^T\theta^*,0\}$, for $j=1,\ldots,n$. As $\cS\subseteq\{0,1\}^n$ and $a^j\ge 0$, for any $x\in\cS$, it holds that
\begin{align}\label{binary}
\max\{x_j-(a^j)^T\theta^*,0\}=\max\{1-(a^j)^T\theta^*,0\}x_j.
\end{align}
Thus, we may write \raf{r-2} as
\begin{alignat}{3}
\label{r-3}
\quad& \displaystyle \OPT_{R}=\OPT_{R}(\theta^*) := \min_{x\in \cS} \quad c(\theta^*)^Tx+b^T \theta^*,
\end{alignat}
where, for any $\theta\in\RR_+^m$,  $c(\theta)$ is the vector with components $c_j(\theta^*):=c^0_j+d_j\cdot\max\{1-(a^j)^T\theta^*,0\}$, for $j=1,\ldots,n$. Let us consider now the relaxation of \raf{r-3}:
\begin{alignat}{3}
\label{r-3-}
\quad& \displaystyle z_{2}^*=z_{2}^*(\theta^*) := \min_{x\in \cQ} \quad c(\theta^*)^Tx+b^T \theta^*,
\end{alignat}
and let $x^*$ be an optimum solution to this relaxation. (Note that this relaxation is not the same as \raf{erlp-rlx}, and in general, one has $z^*_2>z^*_R$.)
By the existence of an $\alpha$-integrality gap verifier, there exists an $\widehat x\in\cS$ such that $c(\theta^*)^T\widehat x\le \alpha \cdot c(\theta^*)^Tx^*$. Then $\widehat x$ is also an $\alpha$-approximate solution to the robust optimization problem \raf{rcop-1}, as
\begin{align}\label{ee1}
c(\theta^*)^T\widehat x+b^T\theta^* \le \alpha \cdot c(\theta^*)^T x^*+b^T\theta^*\le\alpha\big (c(\theta^*)^T x^*+b^T\theta^*\big)=\alpha\cdot z_{2}^*\le \alpha\cdot \OPT_R. 
\end{align}
Let $y^*$ be the vector with components $ y^*_j:=\max\{x_j-(a^j)^T\theta^*,0\}$, for $j=1,\ldots,n$. Then the tuple $(\widehat x,\theta^*,y^*)$ satisfies \raf{y}, and hence (by {\it weak duality}), $\widehat x^Tu\le z^*(\widehat x)\le b^T\theta^*+d^Ty^*$, for all $u\in\RR_+^n$ satisfying~\raf{Aub}. It follows that for any $c=c^0+u\in\cC$ given by \raf{U2}, 
\begin{align}\label{ee2}
c^T\widehat x =(c^0)^T\widehat x+u^T\widehat x\le (c^0)^T\widehat x+b^T\theta^*+d^Ty^*=c(\theta^*)^T\widehat x+b^T\theta^* \le \alpha \cdot  \OPT_R.
\end{align} 
Note that the function $f(\theta,x):=c(\theta)^Tx+b^T\theta$ is {\it quasi-convex} in $\theta$. Hence, even though we can evaluate $f(\theta):=\min_{x\in\cS}f(\theta,x)$ at any point $\theta>0$, within a factor of $\alpha$ (using the $\alpha$-integrality gap verifier), finding $\theta^*\in\argmin_{\theta\ge 0}f(\theta)$ is generally a hard problem.
The rest of the proof is an approximate version  of the above argument in which we approximately "guess" the value of $\theta^*$; this is done in 3 steps: rounding, discretization, and finally calling the  integrality gap verifier for each enumerated value of $\theta^*$. We describe these steps in more details below.
\medskip
	
\noindent{\bf Rounding and discretization.}~ Let the columns of $A$ be $a^1,\ldots, a^n\in\RR_+^m$, and for $j=1,\ldots,n$, denote by $\beta_j:=\max\big\{\max_{i\in[m]}a_{ij},\frac1{d_j}\}$. Note that, we may assume, w.l.o.g., that $0<\beta_j<+\infty$ for all $j$; otherwise,  we may replace $c^0_j$ by $c^0_j+d_j$ and remove $u_j$ from the set of variables in \raf{ir-1}, and the corresponding dual constraints in \raf{ir-1-d}. Similarly, if $b_i=0$ for some $i\in[m]$, we may remove $\{u_j:~a_{ij}>0\}$ from the set of variables in \raf{ir-1} and the corresponding dual constrains in \raf{ir-1-d}. Thus, we may assume in the following that $b_i>0$ for all $i$, and hence (by scaling) $b=\b1_m$, the $m$-dimensional vector of all ones, and that $d_j>0$ for all $j$. Let $\beta:=\min_j\beta_j$ and $\gamma:=\max_j\beta_j$. 

\begin{claim}\label{cl1}
	For any $x\in\{0,1\}^n\setminus\{0\}$, $\frac1{\gamma}\le z^*(x)\le\frac{m+n}{\beta}$.
\end{claim}	
\begin{proof}
	The upper bound follows from the fact that $(\theta=\frac1{\beta}\b1_m,y=\frac1{\beta}d^{-1})$, where $d^{-1}$ is the vector whose $j$th component is $1/d_j$, is a feasible solution for the dual problem~\raf{ir-1-d}.  To see the lower bound, let $(\theta^*,y^*)$ be an optimal solution to the dual problem~\raf{ir-1-d}. Then for any $j$ such that $x_j=1$, we have
	\begin{align*}
	1\le (a^j)^T\theta ^*+y^*_j\le \beta_j(\b1_m^T\theta^*+d_jy^*_j)\le \gamma(\b1_m^T\theta^*+d^Ty^*).
	\end{align*}
\end{proof}

For $j\in[n]$, let 
\begin{align*}\tilde c_j^0:=\left\{\begin{array}{ll}
\frac{\epsilon}{\gamma n},&\text{ if }c_j^0<\frac{\epsilon}{\gamma n},\\
c_j^0,&\text{ otherwise,}
\end{array}
\right.
\end{align*}
and, for $\theta\in\RR_+^m$, define $\widetilde c(\theta)$ to be the vector with components $\widetilde c_j(\theta):=\tilde c^0_j+d_j\cdot\max\{1-(a^j)^T\theta,0\}$, for $j=1,\ldots,n$.  
Define further
 \begin{alignat}{3}
\label{r-3--}
\quad& \displaystyle \widetilde \OPT_{R}= \min_{x\in \cS}\left\{(\tilde c^0)^Tx+ z^*(x)\right\}.
\end{alignat} 
\begin{claim}\label{cl2}
	$ \OPT_{R}\le\widetilde \OPT_{R}\le (1+\epsilon) \OPT_{R}$.
\end{claim}
\begin{proof}
	Let $x^*$ and $\widetilde x$ be optimal solutions for \raf{r-1} and \raf{r-3--}, respectively. Then 
	\begin{align*}
	\OPT_R&\le (c^0)^T\widetilde x+ z^*(\widetilde x)\le(\tilde c^0)^T\widetilde x+ z^*(\widetilde x)=\widetilde \OPT_R\\
	&\le (\tilde c^0)^Tx^*+ z^*(x^*)=(c^0)^Tx^*+ z^*(x^*)+(\tilde c^0-c^0)^Tx^*\le \OPT_R+\frac\epsilon\gamma\\
	&\le \OPT_R+\epsilon\cdot z^*(x^*)\le (1+\epsilon)\OPT_R.\tag{by Claim~\ref{cl1}}
	\end{align*}	
\end{proof}

\begin{claim}\label{cl3}
		$\widetilde  \OPT_{R}\in[\underline z,\overline z]$, where $\underline z:=\max\big\{\min_{j}c^0_j,\frac{\epsilon}{\gamma n}\big\}$, and $\overline z:=n\cdot\max\big\{\max_jc^0_j,\frac{\epsilon}{\gamma n}\big\}+\frac{m+n}{\beta}.$
\end{claim}
\begin{proof}
	Let $\widetilde x$ be optimal solution for~\raf{r-3--}, and $j\in[n]$ be an index such that $\widetilde x_j=1$ (recall that we assume $0\not\in\cS$). Then, using Claim~\ref{cl1},
	\begin{align*}
	\max\big\{\min_{j}c^0_j,\frac{\epsilon}{\gamma n}\big\}\le \tilde c^0_j\le\widetilde  \OPT_{R}= (\tilde c^0)^T\widetilde x+z^*(\widetilde x)\le n\cdot\max\big\{\max_jc^0_j,\frac{\epsilon}{\gamma n}\big\}+\frac{m+n}{\beta}.
	\end{align*}
\end{proof}	
For any $h\in\RR_+$, let $h_\epsilon$ denote a "$(1+\epsilon)$-approximation" of $h$, that is, a number $\widetilde h\in\RR_+$, such that $h<\widetilde h \le (1+\epsilon)h$.   By Claim~\ref{cl3}, we can "guess" a "$(1+\epsilon)$-approximation $\widetilde z_{\epsilon}:=(\widetilde  \OPT_{R})_\epsilon$ by considering the powers $\underline z(1+\epsilon)^{k}$, for $k=0,1,\ldots,\lceil L\rceil $, where 
\begin{align}\label{L}
L&:=\log_{1+\epsilon}\left(\overline z/\underline z\right)=\log_{1+\epsilon}\left(\frac{n\cdot\max\big\{\max_jc^0_j,\frac{\epsilon}{\gamma n}\big\}+\frac{m+n}{\beta}}{\max\big\{\min_{j}c^0_j,\frac{\epsilon}{\gamma n}\big\}}\right)\le \log_{1+\epsilon}\left(\frac{n\cdot\max\big\{\max_jc^0_j,\frac{\epsilon}{\gamma n}\big\}+\frac{m+n}{\beta}}{\frac12\left(\min_{j}c^0_j+\frac{\epsilon}{\gamma n}\right)}\right)\le\\
&\le \log_{1+\epsilon}\left(2n\cdot\max\Big\{\frac{\max_jc^0_j}{\min_{j}c^0_j},\frac{(m+n)\gamma}{\epsilon\beta}\Big\}\right)=O\left( \frac{\log L(A,c^0,d)}{\epsilon}\right).
\end{align}
For $\theta\in\RR_+^m$ we denote by $\theta_\epsilon$ the vector in $\RR_+^m$, whose $i$th component is $(\theta_i)_\epsilon$. Define $\widetilde \OPT_{R}(\theta):= \min_{x\in \cS} \quad \widetilde c(\theta)^Tx+b^T \theta$.
\begin{claim}\label{cl4}
	There exists $\widetilde \theta$ such that
	\begin{align}\label{tht}
	\frac{\epsilon\cdot (\widetilde \OPT_{R})_{\epsilon}}{m}\le\widetilde\theta_i< (\widetilde \OPT_{R})_{\epsilon},\text{ for all $i\in[m]$, }
	\end{align}
	and $\widetilde\OPT_{R}(\widetilde\theta_\epsilon)$ is a $5\epsilon$-approximation of $\OPT_R$.
\end{claim}
\begin{proof}
Let $(\widehat x,y^*,\theta^*)$ be a minimizer of \raf{r-2}, where $c^0$ is replaced by $\tilde c^0$. 	
Define $\widetilde \theta$ as
\begin{align*}
\widetilde\theta_i:=\left\{
\begin{array}{ll}
\frac{\epsilon\cdot  (\widetilde \OPT_{R})_{\epsilon}}{m},&\text{ if }\theta^*_i<\frac{\epsilon\cdot (\widetilde \OPT_{R})_{\epsilon}}{m},\\
\theta^*_i,&\text{ otherwise,}
\end{array}
\right.
\text{ for }i=1,\ldots,m.
\end{align*}	
Since we assume $b=\b1_m$, and $c^0,d\ge0$, we have $\theta_i^*\le    \widetilde \OPT_{R}(\theta^*)=\widetilde\OPT_{R}<(\widetilde \OPT_{R})_{\epsilon}$, and hence, $\widetilde\theta$ satisfies~\raf{tht}. Since $\widetilde \theta_\epsilon>\widetilde \theta\ge\theta^*$ and $A\ge 0$, $(\widehat x,y^*,\widetilde \theta_\epsilon)$ satisfies \raf{cover}. Moreover, 
\begin{align*}
\b1_m^T\widetilde\theta_\epsilon&\le(1+\epsilon)\b1_m^T\widetilde\theta=(1+\epsilon)\big(\b1_m^T\theta^*+\b1_m^T(\widetilde\theta-\theta^*)\big)\le(1+\epsilon)\big(\b1_m^T\theta^*+\epsilon\cdot (\widetilde \OPT_{R})_{\epsilon}\big)\le (1+\epsilon)\b1_m^T\theta^*+\epsilon(1+\epsilon)^2\widetilde \OPT_{R}.
\end{align*}
It follows by Claim~\ref{cl2} that
\begin{align*}
\widetilde\OPT_{R}(\widetilde\theta_\epsilon)&\le (\tilde c^0)^T\widehat x+\b1_m^T\widetilde \theta_\epsilon+d^Ty^*
\le (\tilde c^0)^T\widehat x+(1+\epsilon)\b1_m^T\theta^*+\epsilon(1+\epsilon)^2\widetilde \OPT_{R}+d^Ty^*\\&\le (1+\epsilon+\epsilon(1+\epsilon)^2)\widetilde \OPT_{R}\le(1+5\epsilon)\OPT_R.
\end{align*}
\end{proof}

\medskip

\noindent{\bf Calling the integrality gap verifier.}~It follows from Claim~\ref{cl4} that we can guess an $\epsilon$-approximation $\theta:=\theta_\epsilon$ of $\widetilde\theta$ by considering, for each component $\theta_i$, the powers $\frac{\epsilon\cdot (\widetilde \OPT_{R})_\epsilon}m(1+\epsilon)^{\ell}$, for $\ell=0,1,\ldots,\lceil L'\rceil $ and a guessed value of $\widetilde \OPT_{R}$, where 
$L':=\log\frac{(1+\epsilon)m}{\epsilon}$. By Claim~\ref{cl3}, the total number of possible guesses is
$$
O\left(\frac{\log L(A,c^0,d)}{\epsilon}\right)\times \left(\log\frac{(1+\epsilon)m}{\epsilon}\right)^m.
$$
For each such guess $\theta$, we solve the convex relaxation $\tilde z^*_{R}(\theta):=\min_{x\in\cQ}\widetilde c(\theta)^Tx+\b1_m^T\theta$ to find a minimizer $x^*(\theta)$, and then call the integrality gap verifier on $(\widetilde c(\theta),x^*(\theta))$ to get an integral vector $\widehat x(\theta)\in\cS$ such that $\widetilde c(\theta)^T\widehat x(\theta)\le\alpha\cdot \widetilde c(\theta)^Tx^*(\theta)$. Let $\widehat \theta$ be a minimizer of $\widetilde c(\theta)^T\widehat x(\theta)+\b1_m^T\theta$ over all guesses $\theta$, and write $\widehat x:=x(\widehat\theta)$. Then, similar to \raf{ee1}, it follows that 
\begin{align*}
\widetilde c(\widehat\theta)^T\widehat x+\b1_m^T\widehat\theta&\le\widetilde c(\widetilde\theta_\epsilon)^T\widehat x(\widetilde\theta_\epsilon)+\b1_m^T\widetilde\theta_\epsilon\le\alpha( \widetilde c(\widetilde\theta_\epsilon)^T x^*(\widetilde\theta_\epsilon)+\b1_m^T\widetilde\theta_\epsilon)=\alpha\cdot\tilde z^*_{R}(\widetilde\theta_\epsilon)\le\alpha\cdot\widetilde \OPT_{R}(\widetilde\theta_\epsilon)\le\alpha(1+5\epsilon)\OPT_R,
\end{align*}
where the last inequality follows by Claim~\ref{cl4}.  Let $\widehat y$ be the vector with components $\widehat y_j:=\max\{\widehat x_j-(a^j)^T\widehat\theta,0\}$, for $j=1,\ldots,n$. Then $(\widehat x,\widehat \theta,\widehat y)$ satisfies \raf{y}, and hence, $\widehat x^Tu\le z^*(\widehat x)\le\b1_m^T\widehat\theta+d^T\widehat y$, for all $u\in\RR_+^n$ satisfying~\raf{Aub}. It follows, as in \raf{ee2}, that for any $c=c^0+u\in\cC$, given by \raf{U2}, 
\begin{align}\label{ee2---}
c^T\widehat x =(c^0)^T\widehat x+u^T\widehat x\le (\tilde c^0)^T\widehat x+\b1^T_m\widehat\theta+d^T\widehat y=\widetilde c(\widehat\theta)^T\widehat x+\b1_m^T\widehat\theta \le \alpha (1+5\epsilon)\OPT_{R}.
\end{align} 
\end{proof}
We remark that {\it any $\alpha$-approximation algorithm} can be used instead of an $\alpha$-integrality gap verifier in Theorem~\ref{PU-thm}. Note also that, if both $\frac{\max_jc^0_j}{\min_{j}c^0_j}$ and $\frac{\max_jd_j}{\min_{j}d_j}$ are bounded by $\poly(n)$, then Theorem~\ref{PU-thm} requires only $\polylog(n)$ number of calls to the the integraliy gap verifier, which is an {\it exponential} improvement over the result in \cite{BS03} in such a case.

\medskip

A set $\cS\subseteq\{0,1\}^n$ is said to be {\it covering} if $x\in\cS$ implies that $y\in\cS$ for any $y\ge x$. For instance, if the set $\cS$ represents subgraphs (say, as edge sets) of a given graph satisfying a certain {\it monotone} property (such as connectivity or containment), then $\cS$ is covering.  
Theorem~\ref{PU-thm} gives a reduction from an $\alpha$-integrality gap verifier to $(1+\epsilon)\alpha$-deterministically robust approximation algorithm assuming $m=O(1)$. When $m$ is not a constant, and the set $\cS$ is of the covering type, we have the following result.
\begin{theorem}\label{PU2-thm}
Consider the DO problem \raf{cop-1}, when the set $\cS\subseteq\{0,1\}^n$ is a covering set and the uncertainty set $\cC$ is given by \raf{U2}. Then, there is an $\big(\alpha+2\sqrt{\frac{\alpha\gamma n}{\beta}}\big)$-deterministically robust approximation algorithm for the robust version~\raf{rcop-1}, which can be obtained by a polynomial number of calls to an $\alpha$-integrality gap verifier for the nominal problem \raf{cop-1}. 	
\end{theorem}
\begin{proof}
	We use a {\it dual-fitting} argument~\cite[Chapter 13]{Vaz}.
	Let $z^*_R$ be the value of the relaxation for \raf{r-1}, that is,
	 \begin{alignat}{3}
	\label{rlx-1}
	\quad& \displaystyle z_{R}^* = \min_{x\in \cQ}\left\{(c^0)^Tx+ \max_{u\in\RR_+^n:~ Au\le \b1_m}x^Tu\right\},
	\end{alignat}
	where we assume (w.l.o.g.), for ease of presentation, that the constraint $u\le d$ has already been included in the set of constraints given by $Au\le b$, and that $b:=\b1_m$.
	As in~\raf{r-2}, we can rewrite \raf{rlp-rlx} as
	\begin{alignat}{3}
	\label{rlx-2}
	\quad& \displaystyle z_{R}^* = \min \quad (c^0)^Tx+b^T \theta\\
	\text{s.t.}\quad & \displaystyle A^T\theta\ge x, \label{cover-}\\
	\quad& \theta\in\RR_+^m,\nonumber\\
	\quad& x\in \cQ.\nonumber
	\end{alignat}
	Let $(x^*,\theta^*)$ be an optimal solution for the LP~\raf{rlx-2}. 
	We first call the algorithm in Theorem~\ref{CV-LP} to get a dominated convex combination as in \raf{cvx-comb-LP}, with a polynomially sized set $\cX:=\{x\in\cS:~\mu_x>0\}$. Let $\tau\in(0,1)$ be a number to be chosen later, and define $J:=\{j\in[n]:~x_j^*\ge \tau\}$.  
	\begin{claim}\label{cl5}
		There exists $x\in\cX$ such that $\sum_{j\not\in J}c^0_jx_j+\frac1\beta\sum_{j\not \in J}x_j\le\alpha \big(\sum_{j\not\in J}c^0_jx_j^*+\frac1\beta\sum_{j\not\in J}x^*_j\big)$.
	\end{claim}	
	\begin{proof}
	As $\sum_{x\in\cX}\mu_x=1$, these convex multipliers define a {\it probability distribution} over $\cX$. Let $\widehat x\in\cX$ be selected according to this distribution. Then  
	$\EE[\widehat x]=\sum_{x\in\cX}\mu_xx\leq\alpha x^*$, and by linearity of expectation, $\EE[\sum_{j\not\in J}c^0_j\widehat x_j+\frac1\beta\sum_{j\not \in J}\widehat x_j]\le\alpha \sum_{j\not\in J}c^0_jx_j^*+\frac\alpha\beta \sum_{j\not \in J}x^*_j$. 
	The claim follows.
	\end{proof}
    Let $x\in\cX$ be a vector chosen to satisfy the condition in Claim~\ref{cl5}. We define the rounded vector $\widehat x$ as follows:
    \begin{align*}
    \widehat x_j:=\left\{\begin{array}{ll}
    	x_j, &\text{ if } j\not\in J,\\ 
    	1, &\text{ if } j\in J. 
    \end{array}
    \right.
    \end{align*}
    Note that $\widehat x\in\cS$ since $\cS$ is covering. Now, we define the corresponding dual solution $\widehat\theta$. For $j\not\in J$, define  $i(j)$ to be the smallest $i\in[m]$ such that $i\in\argmax_{i'} a_{i'j}$.
    Next, define $\widehat \theta\in\RR_+^m$ as follows: 
    \begin{align*}
    \widehat \theta_i:= \frac1\tau\theta_i^*+\frac1\beta\sum_{j\not\in J:~i=i(j)}x_j\qquad\text{ for }i=1,\ldots,m.
    \end{align*}
   
    \begin{claim}\label{cl6}
    	$(\widehat x,\widehat\theta)$ is feasible for~\raf{rlx-2} and $(c^0)^T\widehat x+\b1_m^T\widehat\theta\le\left(\alpha+\frac1\tau+\frac{\alpha\gamma \tau n}{\beta}\right)\OPT_R.$
    \end{claim}	
    \begin{proof}
    	First, we show feasibility of $(\widehat x,\widehat\theta)$ for~\raf{rlx-2}.  If $j\in J$, then $(a^j)^T\theta^*\ge x_j^*\ge\tau$, and hence, $(a^j)^T\widehat\theta\ge \frac1\tau (a^j)^T\theta^*\ge 1=\widehat x_j$. On the other hand, if $j\not\in J$ and $\widehat x_j=x_j=1$, then $(a^j)^T\widehat\theta\ge\frac1\beta a_{ij}\sum_{j'\not\in J:~i=i(j')} x_{j'}\ge x_j=1$, where $i:=i(j)$.
    	
    	By definition of $\widehat x$, $\widehat\theta$ and $J$,
       \begin{align*}
          (c^0)^T\widehat x+\b1_m^T\widehat\theta&\le\sum_{j\in J}c^0_j+\sum_{j\not\in J}c^0_jx_j+\frac1\tau\b1_m^T\theta^*+\frac1\beta\sum_i\sum_{j\not\in J:~i=i(j)}x_j\\
          &=\sum_{j\in J}c^0_j+\sum_{j\not\in J}c^0_jx_j+\frac1\tau\b1_m^T\theta^*+\frac1\beta\sum_{j\not\in J}x_j\\
          &\le \frac1\tau\sum_{j\in J}c^0_jx_j^*+\alpha \sum_{j\not\in J}c^0_jx_j^*+\frac1\tau\b1_m^T\theta^*+\frac{\alpha}\beta \sum_{j\not\in J}x^*_j\tag{by the choice of $x$}\\
          &\le\max\{\frac1\tau,\alpha\}\left[(c^0)^Tx^*+\b1_m^T\theta^*\right]+\frac{\alpha \tau n}{\beta}\tag{by the definition of  $J$}\\
          &\le \max\{\frac1\tau,\alpha\}z_{R}^*+\frac{\alpha\gamma \tau n}{\beta}\OPT_R\tag{by Claim~\ref{cl1}}\\&\le\left(\max\{\frac1\tau,\alpha\}+\frac{\alpha\gamma \tau n}{\beta}\right)\OPT_R\le \left(\alpha+\frac1\tau+\frac{\alpha\gamma \tau n}{\beta}\right)\OPT_R.
       \end{align*}
    \end{proof}
 It follows that for any $c=c^0+u\in\cC$, given by \raf{U2}, 
\begin{align}\label{ee2-}
c^T\widehat x =(c^0)^T\widehat x+u^T\widehat x\le (c^0)^T\widehat x+z^*(\widehat x)\le(c^0)^T\widehat x+\b1^T_m\widehat\theta\le \left(\alpha+\frac1\tau+\frac{\alpha\gamma \tau n}{\beta}\right) \OPT_{R}.
\end{align} 
The theorem finally follows by choosing $\tau:=\sqrt{\frac{\beta}{\alpha\gamma n}}$.
\end{proof}	
\begin{remark}\label{r2}
Note in the proof of Theorem~\ref{PU2-thm} that we use {\it strong LP duality} in deriving~\raf{rlx-2}, while only weak duality is used in~\raf{ee2-}. We also note that one does not actually need to compute the dual solution $\theta^*$, but it is only used to obtain a proof of approximate optimality of the integral solution $\widehat x$. 
\end{remark}

\subsection{Robust-with-high-probability Approximation Algorithm for Polyhedral Uncertainty}
Next, we consider the case when the uncertainty set $\cC$ is given by~\raf{Affine} and $\cD=\{\delta\in\RR^{k}_+:A\delta\le b\}$. Let $\beta:=\min_j \max_i a_{ij}$ and $\gamma:=\max_{i,j}a_{ij}$, $c_{\min}:=\min_{r\ne 0,j:~c^r_{j}>0}c^r_{j}$,$c_{\max}^r:=\max_{j}c^r_{j}$ and $c_{\max}:=\max_{r\ne 0}c^r_{\max}$.
\begin{theorem}\label{PU5-thm}
	Consider the DO problem~\raf{cop-1}, when $\cS\subseteq\{0.1\}^n$ and the uncertainty set $\cC$ is given by \raf{Affine} and $\cD=\{\delta\in\RR^{k}_+:A\delta\le b\}$. Then, there is an $O\big(\alpha\sqrt{k \log(k)\frac\gamma\beta\frac{c_{\max}}{c_{\min}}}\big)$-robust-with-high-probability approximation algorithm for the robust version~\raf{rcop-1}, which can be obtained by a polynomial number of calls to an $\alpha$-\Ancrr\ integrality gap verifier for the nominal problem~\raf{cop-1}. 	
\end{theorem}
\begin{proof}
	Assume the availability of an $\alpha$-\Ancrr\ integrality gap verifier for the nominal problem \raf{cop-1}.  We assume w.l.o.g. that $b=\b1_m$.
	Note that the robust DO problem~\raf{rcop-1} in this case takes the form: 
	\begin{alignat}{3}
	\label{r-5}
	\quad& \displaystyle \OPT_{R} = \min_{x\in \cS}\left\{(c^0)^Tx+ \max_{\delta\in\RR_+^k:~ A\delta\le\b1_m}x^T\bC\delta\right\},
	\end{alignat}
	where $\bC\in\RR_+^{n\times k}$ is the matrix whose columns are $c^1,\ldots,c^k$.  
	Let us consider the inner maximization problem in~\raf{r-5} and its dual (for a given $x\in\{0,1\}^n$):
	
	{\centering \hspace*{-18pt}
		\begin{minipage}[t]{.47\textwidth}
			\begin{alignat}{3}
			\label{ir-5}
			\quad& \displaystyle z^*(x) = \max \quad x^T\bC\delta\\
			\text{s.t.}\quad & \displaystyle A\delta\le \b1_m, \label{Adb}\\ 
			&\delta\in\RR_+^k\nonumber
			\end{alignat}
		\end{minipage}
		\,\,\, \rule[-14ex]{1pt}{14ex}
		\begin{minipage}[t]{0.47\textwidth}
			\begin{alignat}{3}
			\label{ir-5-d}
			\quad& \displaystyle z^*(x) = \min \quad \b1_m^T \theta\\
			\text{s.t.}\quad & \displaystyle A^T\theta\ge \bC^Tx, \label{theta-d}\\
			\quad& \theta\in\RR_+^m.\label{nong-d}
			\end{alignat}
	\end{minipage}}

\noindent Note that if $\bC^Tx=0$ for $x\in\{0,1\}^n$, then $z^*(x) =0$ and $x_j=0$ for all $j\in J:=\{j\in[n]~|~\exists r\in[k] :~c^r_j>0\}$. Thus, by considering the relaxation~\raf{lprlx} with $c=c^0$ and $\cQ$ replaced by $\cQ':= \{x\in\cQ:~x_j=0~~\forall j\in J\}$, and calling the integrality gap verifier on the obtained optimal fractional solution $x^*$, we can find an integral solution $\widehat x\in\cS$ such that $\EE[(c^0)^T\widehat x]\le \alpha (c^0)^Tx^*$ (or discover that none exist if the relaxation is infeasible). In view of Remark~\ref{r0}, this expectation guarantee can be turned into a high-probability guarantee without sacrificing much the approximation ratio, that is, we can get a solution $\widehat x^0\in\cS$ such that, with probability $1-o(1)$, we have $(c^0)^T\widehat x^0\le (1+\epsilon)(c^0)^Tx^*$, for any given $\epsilon>0$. We will assume therefore in the following that $\bC^Tx\ne 0$ for all $x\in\cS$, as we will return the minimum of the solution obtained under this assumption and $(c^0)^T\widehat x^0$. 
\begin{claim}\label{cl9}
	For any $x\in\{0,1\}^n$ such that $\bC^Tx\neq0$, we have $z^*(x)\ge\frac{c_{\min}}{\gamma}$.
\end{claim}	
\begin{proof}
		Let $(\theta^*,y^*)$ be an optimal solution to the dual problem~\raf{ir-5-d}.  Since $\bC^Tx\ne0$, there exist $r,j$ such that $c^r_j>0$ and $x_j=1$. Then,
		\begin{align*}
		c_{\min}\le c^r_j\le (c^r)^Tx\le (a^r)^T\theta ^*\le \gamma\b1_m^T\theta^*.
		\end{align*}
		The claim follows.
\end{proof}
	
Let $z^*_R$ be the value of the relaxation for \raf{r-5}, that is,
\begin{alignat}{3}
\label{rlx-5}
\quad& \displaystyle z_{R}^* = \min_{x\in \cQ}\left\{(c^0)^Tx+ \max_{\delta\in\RR_+^k:~ A\delta\le\b1_m}x^T\bC\delta\right\}.
\end{alignat}
As in \raf{r-2}, we may rewrite \raf{rlx-5} as
\begin{alignat}{3}
\label{rlx--5}
\quad& \displaystyle z_{R}^* = \min \quad (c^0)^Tx+\b1_m^T \theta\\
\text{s.t.}\quad & \displaystyle A^T\theta\ge \bC^Tx, \label{cover-d}\\
\quad& \theta\in\RR_+^m, x\in \cQ.\nonumber
\end{alignat}
Let $(x^*,\theta^*)$ be an optimal solution for the LP~\raf{rlx--5}. 
We call the $\alpha$-\Ancrr\ integrality gap verifier on $x^*$ to get an $\alpha$-\Ancrr\ $\widehat x\in\cS$. Let $\tau\in(0,1)$ be a number to be chosen later, and define $R:=\{r\in[k]:~(c^r)^Tx^*\ge \tau c_{\max}^r\}$.  
\begin{claim}\label{cl8}
	For $\rho\ge 1$, $\Pr\left[\forall r\in R:~(c^r)^T\widehat x\le(1+\rho)\alpha (a^r)^T\theta^*\right]\ge 1-ke^{-\rho\alpha\tau/3}$.
\end{claim}	
\begin{proof}
	We will use the following extension of Chernoff bound:
	\begin{fact}[\cite{PS97}]\label{f1}
		Let $w\in [0,1]^n$ be a given vector of numbers and $\widehat x\in\{0,1\}^n$ be a vector of random variables. 
		Suppose that~\raf{nc2} holds for all $S\subseteq[n]$, and $\EE[w^T\widehat x] \le \mu$. Then for any $\rho\ge1$, we have $\Pr[w^T\widehat x\ge(1+\rho)\mu]\le  e^{-\mu\rho/3}$.
	\end{fact}
Fix $r\in R$. Applying Fact~\ref{f1} with $w:=\frac{c^r}{c_{\max}^r}\in[0,1]^n$ and noting by property~(i) of an $\alpha$-\Ancrr\ and the feasibility of $x^*$ for~\raf{theta-d} that $\EE[(c^r)^T\widehat x]\le\alpha(c^r)^Tx^*\le\alpha (a^r)^T\theta^*$, we  obtain, for $\rho\ge 1$,
\begin{align*}
\Pr[(c^r)^T\widehat x\ge(1+\rho)\alpha (a^r)^T\theta^*]&\le\Pr\big[(c^r)^T\widehat x\ge(1+\rho)\alpha(c^r)^Tx^*\big]\le e^{-\frac{\rho\alpha(c^r)^Tx^*}{3c_{\max}^r}}\le e^{-\rho\alpha\tau/3}.
\end{align*}
The claim follows by applying a union bound over all $r\in R$.
\end{proof}
For $r\not\in R$, define  $i(r)$ to be the smallest $i\in[m]$ such that $i\in\argmax_{i'} a_{i'r}$.
Let us next choose $\rho:=\frac{6\ln(2k)}{\tau}>1$ and define the dual solution $\widehat \theta\in\RR_+^m$ as follows: 
\begin{align*}
\widehat \theta_i:=
(1+\rho)\alpha\theta_i^*+\frac1\beta\sum_{r\not\in R:~i=i(r)}(c^r)^T\widehat x,\qquad\text{ for }i=1,\ldots,m.
\end{align*}
Let us fix an arbitrary constant $\epsilon\in(0,1)$.
  \begin{claim}\label{cl10}
	With probability $1-o(1)$, $(\widehat x,\widehat\theta)$ is feasible for \raf{rlx--5} and $(c^0)^T\widehat x+\b1_m^T\widehat\theta\le\left((1+\rho)+\frac{(1+\epsilon)\gamma \tau c_{max} k}{\beta c_{\min}}\right)\alpha \OPT_R.$
\end{claim}	
\begin{proof}
	First, we show feasibility of $(\widehat x,\widehat\theta)$.  By Claim~\ref{cl8}, with probability $1-ke^{-\rho\alpha\tau/3}\ge1-\frac1{4k}$, for all $r\in R$, we have $(1+\rho)\alpha (a^r)^T\theta^*\ge (c^r)^T\widehat x$, and hence, $(a^r)^T\widehat\theta\ge(c^r)^T\widehat x$. On the other hand, if $r\not\in R$, then $(a^r)^T\widehat\theta\ge\frac{a_{ij}}\beta\sum_{r'\not\in R:~i=i(r')}(c^{r'})^T\widehat x\ge(c^{r})^T\widehat x$, where $i:=i(r)$.
	
	By property~(i) of an $\alpha$-\Ancrr, we have $\EE[(c^0)^T\widehat x]\le\alpha (c^0)^Tx^*$, and by definition of $R$, we have $\EE[\sum_i\sum_{r\not\in R:~i=i(r)}(c^r)^T\widehat x]\le \alpha\sum_i\sum_{r\not\in R:~i=i(r)}(c^r)^Tx^*\le k\alpha\tau c_{\max}$. Thus, in view of Remark~\ref{r0}, with prob. $1-o(1)$, for any $\epsilon\in(0,1)$, we have  
	$(c^0)^T\widehat x+\frac1\beta\sum_i\sum_{r\not\in R:~i=i(r)}(c^r)^T\widehat x \le(1+\epsilon)\big(\alpha (c^0)^Tx^*+\frac{k\alpha\tau c_{\max}}{\beta}\big)$.
	It follows that, with prob. $1-o(1)$, we have
	\begin{align*}
    (c^0)^T\widehat x+\b1_m^T\widehat\theta&=(c^0)^T\widehat x+(1+\rho)\alpha\b1_m^T\theta^*+\frac1\beta\sum_i\sum_{r\not\in R:~i=i(r)}(c^r)^T\widehat x\\
    &\le(1+\epsilon)\alpha (c^0)^Tx^*+(1+\rho)\alpha\b1_m^T\theta^*+(1+\epsilon)\frac{k\alpha\tau c_{\max}}\beta\\
	&\le (1+\rho)\alpha z^*_R+\frac{(1+\epsilon)\alpha\gamma \tau c_{\max} k}{\beta c_{\min}}\OPT_R\tag{by Claim~\ref{cl9}}\\&\le\left((1+\rho)+\frac{(1+\epsilon)\gamma \tau c_{\max} k}{\beta c_{\min}}\right)\alpha \OPT_R.
	\end{align*}
\end{proof}
It follows from Claim~\ref{cl10} that, With probability $1-o(1)$, for any $c=c^0+\bC\delta\in\cC$, given by \raf{Affine} with  $\cD=\{\delta\in\RR^{k}_+:A\delta\le b\}$, we have
\begin{align}\label{ee2--}
c^T\widehat x =(c^0)^T\widehat x+(\bC\delta)^T\widehat x\le (c^0)^T\widehat x+z^*(\widehat x)\le(c^0)^T\widehat x+\b1^T_m\widehat\theta\le \left((1+\rho)+\frac{(1+\epsilon)\gamma \tau c_{\max} k}{\beta c_{\min}}\right)\alpha \OPT_R.
\end{align} 
The theorem follows by choosing $\tau:=\sqrt{\frac{6\beta\ln(2k)c_{\min}}{(1+\epsilon)\alpha\gamma c_{max} k}}$.
\end{proof}
\subsection{A Deterministic Robust Approximation Algorithm for Ellipsoidal Uncertainty}

Consider the DO problem \raf{cop-1} and its LP-relaxation~\raf{lprlx}, when the uncertainty set $\cC$ is given by the ellipsoid: 
\begin{align}\label{U3}
\cC=\left\{c:=c^0+\sum_{r=1}^k\delta_rc^r\left|~\delta\in\RR^k,~\|D^{-1}\delta\|_2\le 1\right.\right\},
\end{align}
where $c_0,c^1,\ldots,c^k\in\RR_+^n$ are given non-negative vectors, and $D\in\RR^{k\times k}_+$ is a given positive definite matrix. 
\begin{theorem}\label{PU3-thm}
	Consider the DO problem \raf{cop-1} and its relaxation~\raf{rlp-rlx}, when the uncertainty set $\cC$ is given by \raf{U3}, such that $D\bC\ge0$. Then, there is an $O\big(\alpha\sqrt{k}\big)$-deterministically robust approximation algorithm for the robust minimization problem~\raf{rcop-1}, which can be obtained by a polynomial number of calls to an $\alpha$-integrality gap verifier for the nominal problem \raf{cop-1}. 	
\end{theorem}
\begin{proof}
	Let $x^*$ be an optimal solution for the convex relaxation~\raf{rlp-rlx}:
	\begin{alignat}{3}
	\label{rlx-3}
	\quad& \displaystyle z_{R}^* = \min_{x\in \cQ}\left\{(c^0)^Tx+ \max_{w\in\RR^n:~\| w\|_2\le 1}x^T\bC Dw\right\} = \min_{x\in \cQ}\left\{(c^0)^Tx+\|D\bC^Tx\|_2\right\}.
	\end{alignat}  
	Call the algorithm in Theorem~\ref{CV-LP} to get a dominated convex combination as in \raf{cvx-comb-LP}, with a polynomially sized set $\cX:=\{x\in\cS:~\mu_x>0\}$. Choose $\widehat x\in\argmin_{x\in\cX}\big(c_0^Tx+\|DC^Tx\|_2\big)$. Then
	\begin{align*}
	c_0^T\widehat x+\|D\bC^T\widehat x\| &\le \sum_{x\in\cX}\mu_x(c^0)^Tx+\sum_{x\in\cX}\mu_x\|D\bC^Tx\|_2 \le \sum_{x\in\cX}\mu_x(c^0)^Tx+\sum_{x\in\cX}\mu_x\|D\bC^Tx\|_1\\
	&=\sum_{x\in\cX}\mu_x(c^0)^Tx+\sum_{x\in\cX}\mu_x\sum_j(D\bC^Tx)_j= \sum_{x\in\cX}\mu_x(c^0)^Tx+\sum_j\sum_{x\in\cX}\mu_x(D\bC^Tx)_j\tag{$\because D\bC\ge0$}\\
	&\le\alpha(c^0)^Tx^*+\alpha\sum_j(D\bC^Tx^*)_j=\alpha(c^0)^Tx^*+\alpha\|D\bC^Tx^*\|_1\\&\le\alpha(c^0)^Tx^*+\alpha\sqrt n\|D\bC^Tx^*\|_2\le\alpha\sqrt k \cdot z^*_R\le \alpha\sqrt k \cdot \OPT_{R}.
	\end{align*}
	 It follows that for any $c=c^0+\bC\delta\in\cC$, given by \raf{U3}, 
	\begin{align*}
	c^T\widehat x =(c^0)^T\widehat x+(\bC\delta)^T\widehat x=(c^0)^T\widehat x+(D^{-1}\delta)^T(D\bC^T\widehat x)\le(c^0)^T\widehat x+\|D^{-1}\delta\|_2\cdot\|D\bC^T\widehat x\|_2\le(c^0)^T\widehat x+\|D\bC^T\widehat x\|_2 \le\alpha\sqrt k \cdot  \OPT_{R}.
	\end{align*} 
\end{proof}

Next, let us consider the case when the set $\cS\subseteq\{0,1\}^n$ and the uncertainty set $\cC$ is given by 
\begin{align}\label{U4}
\cC=\left\{c:=c^0+u~\left|~u\in\RR_+^n,~\|D^{-1}u\|_2\le 1\right.\right\},
\end{align}
for a given  a given positive definite matrix $D\in\RR^{n\times n}$. 
\begin{theorem}\label{PU4-thm}
	Consider the DO problem \raf{cop-1} and its relaxation~\raf{rlp-rlx}, when the set $\cS\subseteq\{0,1\}^n$ is a covering set and the uncertainty set $\cC$ is given by \raf{U4}, such that $D^{-1}>0$. Then, there is an $O\big(\alpha+\sqrt{\frac{\alpha\lambda_{\max}(D) n}{\lambda_{\min}(D)}}\big)$-deterministically robust approximation algorithm for the robust version~\raf{rcop-1}, which can be obtained by a polynomial number of calls to an $\alpha$-integrality gap verifier for the nominal problem \raf{cop-1}. 	
\end{theorem}
\begin{proof}
	Note that the robust DO problem~\raf{rcop-1} in this case takes the form: 
	\begin{alignat}{3}
	\label{r-4}
	\quad& \displaystyle \OPT_{R} = \min_{x\in \cS}\left\{(c^0)^Tx+ \max_{u\in\RR_+^n:~\|D^{-1}u\|_2\le 1}x^Tu\right\}.
	\end{alignat}
	Let us consider the inner maximization problem in~\raf{r-1}, which can be written as the following {\it semi-infinite} LP (as $D^{-1}>0$), and its dual (for a given $x\in\RR_+^n$):
	
	{\centering \hspace*{-18pt}
		\begin{minipage}[t]{.47\textwidth}
			\begin{alignat}{3}
			\label{ir-4}
			\quad& \displaystyle z^*(x) = \max \quad x^T u\\
			\quad& v^Tu\le 1,~\forall v\in E_+(0,D^{-1}):=\{v\in\RR_+^n:~v^TD^{2}v\le 1\},\nonumber\\
			\quad&  u\ge 0,~u\in\RR^n\nonumber
			\end{alignat}
		\end{minipage}
		\,\,\, \rule[-14ex]{1pt}{14ex}
		\begin{minipage}[t]{0.47\textwidth}
			\begin{alignat}{3}
			\label{ir-4-d}
			\quad& \displaystyle z^*(x) = \min \quad \int_{v\in E_+(0,D^{-1})}\theta(v)dv\\
			\text{s.t.}\quad & \displaystyle \int_{v\in E_+(0,D^{-1})}\theta(v)vdv\ge x, \label{theta}\\
			\quad& \theta:E_+(0,D^{-1})\to\RR_+.\label{nong3}
			\end{alignat}
	\end{minipage}}
	
	\noindent It was shown in \cite{EMN19} that, for any given $x\in\RR_+^n$,  a near-optimal solution for \raf{ir-4} can be obtained in polynomial time, using {\it multiplicative weight updates}, which also produces a near-optimal solution to \raf{ir-4-d}. More precisely, for any $\epsilon>0$,  we can find in $\poly(n,m,\log\frac{\lambda_{\max}(D)}{\lambda_{\min(D)}},\frac1\epsilon)$ time, where $\lambda_{\min}(D)$ and $\lambda_{\max}(D)$  are the minimum and maximum eigenvalues of $D$ respectively, and  vectors $\theta^*\in\RR_+^m$ and a (implicitly described) function $\theta^*:E(0,D^{-1})\to\RR_+$, satisfying \raf{theta} and
	\begin{align}\label{near-opt2}
	\int_{v\in E_+(0,D^{-1})}\theta^*(v)dv\le(1+\epsilon)z^*(x).
	\end{align}
	Note that \raf{near-opt2} implies that {\it strong} duality holds, as we can set $\epsilon\to0$. 
	Thus, we may write~\raf{r-4} as follows:
	\begin{alignat}{3}
	\label{r-4-d}
	\quad& \displaystyle \OPT_{R} = \min \quad (c^0)^Tx+\int_{v\in E_+(0,D^{-1})}\theta(v)dv\\
	\text{s.t.}\quad & \displaystyle \int_{v\in E_+(0,D^{-1})}v\theta(v)dv\ge x,\nonumber\\
	\quad&\theta:E_+(0,D^{-1})\to\RR_+,~x\in\cS.\nonumber
	\end{alignat}
	We will make use of the lower bound in the following claim. 
	\begin{claim}\label{cl7-}
		For any $x\in\{0,1\}^n\setminus\{0\}$, $\lambda_{\min}(D)\le z^*(x)\le\lambda_{\max}(D)\sqrt{n}$.
	\end{claim}	
	\begin{proof}
		First we show the upper bound:
		\begin{align}\label{eee2}
		\displaystyle z^*(x)=\max_{u\ge 0:~v^Tu\le 1,~\forall v\in E_+(0,D^{-1})}x^Tu=\max_{u\ge 0:~\|D^{-1}u\|_2\le 1}x^Tu\le\max_{u\ge 0:~\|D^{-1}u\|_2\le 1}\b1_n^Tu\le \|D\b1_n\|_2\le\lambda_{\max}(D)\sqrt{n}.
		\end{align}
		To see the lower bound, let $\theta^*$ be an optimal solution to the dual problem~\raf{ir-4-d}. Then for any $j$ such that $x_j=1$, we have
		\begin{align*}
		1\le \int_{v\in E_+(0,D^{-1})}v_j\theta^*(v)dv\le\lambda_{\max}(D^{-1})\int_{v\in E_+(0,D^{-1})}\theta^*(v)dv=\lambda_{\max}(D^{-1})z^*(x).
		\end{align*}
		The claim follows.
	\end{proof}
	Let $x^*\in\cQ$ be an optimal solution for the convex relaxation~\raf{rlp-rlx}: 
	\begin{alignat}{3}
	\label{rlx-4}
	\quad& \displaystyle z_{R}^* = \min_{x\in \cQ}\left\{(c^0)^Tx+ \max_{w\in\RR^n_+:~\|D^{-1}u\|_2\le 1}u^Tx\right\},
	\end{alignat}
	and $\theta^*$ be a corresponding dual solution. 
	(As in Remark~\ref{r2}, we do not actually need to compute the dual solution, but we use its existence to bound the rounded integral solution.) 
	We first call the algorithm in Theorem~\ref{CV-LP} to get a dominated convex combination as in \raf{cvx-comb-LP}, with a polynomially sized set $\cX:=\{x\in\cS:~\mu_x>0\}$. Let $\tau\in(0,1)$ be a number to be chosen later, and define $J:=\{j\in[n]:~x_j^*\ge \tau\}$.  
	Similar to Claim~\ref{cl5}, we can prove the following.
	\begin{claim}\label{cl5-}
		There exists $x\in\cX$ such that $\sum_{j\not\in J}c^0_jx_j+\lambda_{\max(D)}\sum_{j\not \in J}x_j\le\alpha \big(\sum_{j\not\in J}c^0_jx_j^*+\lambda_{\max(D)}\sum_{j\not\in J}x^*_j\big)$.
	\end{claim}	
	Let $x\in\cX$ be a vector chosen to satisfy the condition in Claim~\ref{cl5-}. We define the rounded vector $\widehat x$ as follows:
	\begin{align*}
	\widehat x_j:=\left\{\begin{array}{ll}
	x_j, &\text{ if } j\not\in J,\\ 
	1, &\text{ if } j\in J. 
	\end{array}
	\right.
	\end{align*}
	Note that $\widehat x\in\cS$ since $\cS$ is covering. 
	We next define $\widehat\theta:E_+(0,D^{-1})\to\RR_+$ as follows: 
	\begin{align*}
	\widehat \theta(v)&:=\frac1\tau\theta^*(v)+\lambda_{\max(D)}\sum_{v'\in T}\delta_n(v-v'), 
	\end{align*}
	where $T:=\{v\in E_+(0,D^{-1}):~\exists j\not\in J \text{ s.t. }x_j=1\text{ and }v_j=\max_{v'\in E_+(0,D^{-1})} v_j'\}$ and $\delta_n:\RR^n\to\RR$ is the $n$-dimensional {\it Dirac delta function} satisfying $\delta_n(v)=0$ for all $v\neq 0$ and $\int_{v\in\RR^n}\delta_n(v)dv=1$. 
	\begin{claim}\label{cl6-}
		$(\widehat x,\widehat\theta)$ is feasible for \raf{r-4-d} and $(c^0)^T\widehat x+ \int_{v\in E_+(0,D^{-1})}\widehat\theta(v)dv\le\left(\alpha+\frac1\tau+\frac{\alpha\lambda_{\max(D)} \tau n}{{\lambda_{\min(D)}}}\right)\OPT_R.$
	\end{claim}	
	\begin{proof}
		First, we show feasibility of $(\widehat x,\widehat\theta)$.  If $j\in J$, then $\int_{v\in E_+(0,D^{-1})}v_j\theta^*(v)dv\ge x_j^*\ge\tau$, and hence, $\int_{v\in E_+(0,D^{-1})}v_j\widehat\theta(v)dv\ge \frac1\tau (\int_{v\in E_+(0,D^{-1})}v_j\theta^*(v)dv)\ge 1=\widehat x_j$. On the other hand, if $j\not\in J$ and $\widehat x_j=x_j=1$, then 
		\begin{align*}
		\int_{v\in E_+(0,D^{-1})}v_j\widehat\theta(v)dv&\ge \lambda_{\max(D)}\sum_{v'\in T}\int_{v\in E_+(0,D^{-1})}v_j\delta_n(v-v')dv\ge\lambda_{\max(D)}\max_{v\in E_+(0,D^{-1})} v_j\\&=\lambda_{\max(D)}\max_{\|w\|\le1,~D^{-1}w\ge 0}(\b1_n^j)^TD^{-1}w\ge\lambda_{\max(D)}(\b1_n^j)^TD^{-1}\b1_n^j\ge\lambda_{\max(D)}\lambda_{\min(D^{-1})}=1,
		\end{align*}
		where $\b1_j^n$ is the $j$ unit vector of dimension $n$.
		Also, by definition of $\widehat x$, $\widehat\theta$ and $J$,
		\begin{align*}
		(c^0)^T\widehat x+\int_{v\in E_+(0,D^{-1})}\widehat\theta(v)dv&\le\sum_{j\in J}c^0_j+\sum_{j\not\in J}c^0_jx_j+\frac1\tau\int_{v\in E_+(0,D^{-1})}\theta^*(v)dv+\lambda_{\max(D)}\sum_{v'\in T}\int_{v\in E_+(0,D^{-1})}\delta_n(v-v')dv\\
		(c^0)^T\widehat x+\int_{v\in E_+(0,D^{-1})}\widehat\theta(v)dv&\le\sum_{j\in J}c^0_j+\sum_{j\not\in J}c^0_jx_j+\frac1\tau\int_{v\in E_+(0,D^{-1})}\theta^*(v)dv+\lambda_{\max(D)}\sum_{j\not\in J}x_j\\
		&\le \frac1\tau\sum_{j\in J}c^0_jx_j^*+\alpha \sum_{j\not\in J}c^0_jx_j^*+\frac1\tau\int_{v\in E_+(0,D^{-1})}\theta^*(v)dv+\alpha\lambda_{\max(D)} \sum_{j\not\in J}x^*_j\tag{by  Claim~\ref{cl5-}}\\
		&\le\max\{\frac1\tau,\alpha\}\left[(c^0)^Tx^*+\int_{v\in E_+(0,D^{-1})}\theta^*(v)dv\right]+\alpha \tau n\lambda_{\max(D)}\\
		&\le \max\{\frac1\tau,\alpha\}\left[(c^0)^Tx^*+\int_{v\in E_+(0,D^{-1})}\theta^*(v)dv\right]+\frac{\alpha\lambda_{\max(D)} \tau n}{{\lambda_{\min(D)}}}\OPT_R\tag{by Claim~\ref{cl7-}}\\&\le\left(\max\{\frac1\tau,\alpha\}+\frac{\alpha\lambda_{\max(D)}\tau n}{{\lambda_{\min(D)}}}\right)\OPT_R\le \left(\alpha+\frac1\tau+\frac{\alpha{\lambda_{\max(D)}} \tau n}{{\lambda_{\min(D)}}}\right)\OPT_R.
		\end{align*}
	\end{proof}
	It follows that for any $c=c^0+u\in\cC$, given by \raf{U4}, 
	\begin{align}\label{ee3}
	c^T\widehat x  =(c^0)^T\widehat x+u^T\widehat x\le (c^0)^T\widehat x+ \int_{v\in E_+(0,D^{-1})}\widehat\theta(v)dv+\le \left(\alpha+\frac1\tau+\frac{\alpha\lambda_{\max(D)} \tau n}{\lambda_{\min(D)}}\right) \OPT_{R}.
	\end{align} 
	The theorem follows by choosing $\tau:=\sqrt{\frac{{\lambda_{\min(D)}}}{\alpha{\lambda_{\max(D)}} n}}$.
\end{proof}

\subsection{Robust-with-high-probability Approximation Algorithm for Ellipsoidal Uncertainty}
By arguments similar to the ones used to prove Theorems~\ref{PU5-thm} and ~\ref{PU4-thm} we obtain the following result.
\begin{theorem}\label{PU6-thm}
	Consider the DO problem~\raf{cop-1}, when $\cS\subseteq\{0.1\}^n$ and the uncertainty set $\cC$ is given by \raf{Affine} and $\cD=\{\delta\in\RR^{k}_+:\|D^{-1}\delta\|_2\le 1\}$. Then, there is an $O\big(\alpha\sqrt{k \log(k)\frac{\lambda_{\max(D)}}{\lambda_{\min(D)}}\frac{c_{\max}}{c_{\min}}}\big)$-robust-with-high-probability approximation algorithm for the robust version~\raf{rcop-1}, which can be obtained by a polynomial number of calls to an $\alpha$-\Ancrr\ integrality gap verifier for the nominal problem~\raf{cop-1}. 	
\end{theorem}

\begin{remark}\label{r3}
	We may also consider the case when the set $\cS\subseteq\{0,1\}^n$  and the uncertainty set $\cC$ is given by
	\begin{align}\label{U4-}
	\cC=\left\{c:=c^0+u~\left|~u\in\RR_+^n,~\|D^{-1}u\|_2\le 1,~u\le d,\quad Au\le \b1_m\right.\right\},
	\end{align}
	for a given non-negative matrix $A\in\RR_+^{m\times n}$. Using similar techniques, a bound generalizing the results of Theorems~\ref{PU2-thm} and \ref{PU4-thm}  can be obtained.
\end{remark}	


\end{document}